\documentclass[11pt]{amsart}
\usepackage{amssymb,color}
\numberwithin{equation}{section} 

\usepackage{graphicx}

\newcommand{\bea}{\begin{eqnarray}}
\newcommand{\eea}{\end{eqnarray}}
\newcommand{\ba}{\begin{array}}
\newcommand{\ea}{\end{array}}
\newcommand{\edc}{\end{document}}
\newcommand{\bc}{\begin{center}}
\newcommand{\ec}{\end{center}}
\newcommand{\be}{\begin{equation}}
\newcommand{\ee}{\end{equation}}

\def\bc{{\mathbb C}}

\def\bn{{\mathbb N}}

\def\b{\beta}

\def\G{\Gamma}

\def\m{\mu}

\def\s{\sigma}

\def\w{\omega}
\def\Om{\Omega}

\def\h{{\mathbf{h}}}

\newtheorem{thm}{Theorem}[section]
\newtheorem{lem}[thm]{Lemma}

\newtheorem{prop}[thm]{Proposition}

\theoremstyle{remark}

\date{\today}
\begin{document}

\title[Gibbs measures]
{Gibbs measures and free energies of Ising-Vannimenus Model on the
Cayley tree}
\author{Farrukh Mukhamedov}
\address{Farrukh Mukhamedov\\
Department of Computational \& Theoretical Sciences\\
Faculty of Science, International Islamic University Malaysia\\
P.O. Box, 141, 25710, Kuantan\\
Pahang, Malaysia} \email{{\tt far75m@yandex.ru;} {\tt
farrukh\_m@iium.edu.my}}

\author{Hasan Ak\i n}
\address{Hasan Ak\i n, Department of Mathematics, Faculty of Education,
 Zirve University, Kizilhisar Campus, Gaziantep, 27260, Turkey}
\email{{\tt akinhasan25@gmail.com;}   {\tt
hasan.akin@zirve.edu.tr}}

\begin{abstract}
In this paper, we consider Ising-Vannimenus model on a Cayley tree
for order two with competing nearest-neighbor, prolonged
next-nearest neighbor interactions. We stress that the mentioned
model was investigated only numerically, without rigorous
(mathematical) proofs. One of the main point of this paper is to
propose a measure-theoretical approach the considered model. We
find certain conditions for the existence of Gibbs measures
corresponding to the model. Then we establish the existence of the
phase transition. Moreover, the free energies of the found Gibbs
measures are calculated.
\end{abstract}

\maketitle
\section{introduction}
Gibbs measure is ne of the central objects of equilibrium statistical mechanic, a branch of probability theory that takes its origin from Boltzmann [2]. Also, one of the main problems of the Statistical Physics is to describe all Gibbs measures corresponding to the given Hamiltonian. It is well known that such measures form a nonempty convex compact subset in the set of all probabilistic measures. The purpose of this paper is to investigate Gibbs measures of the Ising model \cite{V} with ternary prolonged and nearest neighbor interactions on Bethe lattice of order two and to describe its extreme elements (pure phases). In \cite{GTA}, we have studied phase diagram and extreme Gibbs measures of the Ising model on a Cayley tree in the presence of competing binary and ternary interactions. In \cite{NHSS}, we have obtained the extreme Gibbs measures of the Ising-Vannimenus (IV) model \cite{V} without proof a number of theorems related to the description of the general structure of extreme Gibbs distributions of the Ising model on a Cayley. we studied Gibbs states (phases) that correspond in probability theory to what are called Markov chains with memory length 2 by using the method in \cite{NHSS}. In this paper, we combine the results obtained in \cite{NHSS1} and \cite{NHSS}. In \cite{MDA}, the authors study the existence
of $p$-adic quasi Gibbs measures by means of investigation of the obtained equations via methods of $p$-adic analysis. Note that the methods used in \cite{MDA} are not valid in a real setting. Therefore, in this paper we will give a rigorous description of Gibbs measures corresponding to Ising-Vannimnus in the real setting.

Recently, Ganfoldo et al \cite{GHRR} have obtained some explicit formulae of the free energies (and entropies) according to boundary conditions (b.c.) for the Ising model on the Cayley tree. Also, Ganfoldo et al \cite{GRS} study the Ising model on a Cayley tree. They show that a wide class of new extreme Gibbs states is exhibited.

Until now, many researchers have investigated Gibbs measure with
memory one over Cayley tree. In this paper, we consider Ising
model on a Cayley tree for order two with competing
nearest-neighbor, prolonged next-nearest neighbor interactions. We
stress that the mentioned model was investigated only numerically,
without rigorous (mathematical) proofs. We propose a rigorous
measure-theoretical approach to investigate a Ising-Vanniminus
model on a Cayley tree of order two.  We find certain conditions
for the existence of Gibbs measures corresponding to the model.
Then we establish the existence of the phase transition. We
investigate the behavior of free energy at these new transition
points. We present some explicit formulae of the free energies
corresponding to boundary conditions for the Ising-Vannimenus
model on the Cayley tree of order two.

\section{Preliminaries}
\subsection{Vannimenus Model with competing interactions on Cayley tree}

A Cayley tree $\Gamma^k$ of order $k\geq 1 $ is an infinite tree, i.e., a graph without cycles with
exactly $ k+1 $ edges issuing from each vertex. Let us denote the
Cayley tree as $\Gamma^k=(V, \Lambda,i)$ where $V$ is the set of
vertices of $ \Gamma^k$, $\Lambda$ is the set of edges, and $i$ the
incidence function which associates with each edge $l\in L$ its
endpoints. Two vertices $x$ and $y$, $x,y \in V$ are called {\it
nearest-neighbors} if there exists an edge $l\in\Lambda$ connecting
them, which is denoted by $l=<x,y>$. The distance $d(x,y), x,y\in
V$, on the Cayley tree $\Gamma^k$, is the number of edges in the
shortest path from $x$ to $y$. Two vertices $x,y\in V^0$ are called
{\it the next-nearest-neighbours} if $d(x,y)=2$.
Next-nearest-neighbour vertices $x$ and $y$ are called {\it
prolonged next-nearest-neighbours} if they belong to the same branch
which is denoted by $\widetilde{>x,y<}$. Let spin variables
$\sigma(x), x\in V,$ take values $\pm 1$. Then the Vannimenus  model
with competing nearest-neighbours and next-nearest-neighbours binary
interactions is defined by the following Hamiltonian
\begin{equation}\label{hm}
H(\sigma)=-J_p\sum_{\widetilde{>x,y<}}\sigma(x)\sigma(y) - J
\sum_{<x,y>}\sigma(x)\sigma(y),
\end{equation}
where the sum in the first term ranges all prolonged
next-nearest-neighbours and  the sum in the second   term ranges all
nearest-neighbours. Here $J_p,J\in {R}$ are coupling constants.

The distance $d(x,y), x,y\in V$, on
the Cayley tree $\Gamma^k$, is the number of edges in the shortest
path from $x$ to $y$. For a fixed $x^0\in V$ we set
$$
W_n=\{x\in V| d(x,x^0)=n\},  \\ V_n=\{x\in V| d(x,x^0)\leq n\}
$$
and $L_n$ denotes the set of edges in $V_n$. The fixed vertex
$x^0$ is called the $0$-th level and the vertices in $W_n$ are
called the $n$-th level and for $x\in W_n$ let
$$
S(x)=\{y\in W_{n-1}| d(x,y)=1\},
$$
be the set of direct successors of $x$.

Note in \cite{V} it is assumed that $J>0$ and $J_p<0.$ Below we
consider model (1) with arbitrary sign of the coupling constants.

As usual, one can introduce the notions of Gibbs distribution of
this model, limiting Gibbs distribution, pure phase ( extreme
Gibbs distribution ), etc (see \cite{D1}, \cite{FV},\cite{G},
\cite{P}).

\section{Translation-invariant Gibbs measures of the Ising-Vannimenus model}

In this section we define a notion of Gibbs measure corresponding to
the Ising-Vannimenus model in a general setting, i.e. for arbitrary
nearest-neighbor models (see \cite{V}).
We propose a new kind of construction of Gibbs measures corresponding to the Ising-Vannimenus model.

Let $\Phi=\{-1,+1\}$ ($\Phi$ is called a {\it state space}) and is assigned to the vertices of the tree
$\G^k_+=(V,\Lambda)$. A configuration $\s$ on $V$ is then defined as
a function $x\in V\to\s(x)\in\Phi$; in a similar manner one defines
configurations $\s_n$ and $\w$ on $V_n$ and $W_n$, respectively. The
set of all configurations on $V$ (resp. $V_n$, $W_n$) coincides with
$\Omega=\Phi^{V}$ (resp. $\Omega_{V_n}=\Phi^{V_n},\ \
\Omega_{W_n}=\Phi^{W_n}$). One can see that
$\Om_{V_n}=\Om_{V_{n-1}}\times\Om_{W_n}$. Using this, for given
configurations $\s_{n-1}\in\Om_{V_{n-1}}$ and $\w\in\Om_{W_{n}}$ we
define their concatenations  by
$$
(\s_{n-1}\vee\w)(x)= \left\{
\begin{array}{ll}
\s_{n-1}(x), \ \ \textrm{if} \ \  x\in V_{n-1},\\
\w(x), \ \ \ \ \ \ \textrm{if} \ \ x\in W_n.\\
\end{array}
\right.
$$
It is clear that $\s_{n-1}\vee\w\in \Om_{V_n}$.

For the Ising model with spin values in $\Phi=\{-1,+1\}$, the relevant Hamiltonian with competing
nearest-neighbor, prolonged next-nearest-neighbor and two-level
triple interactions has the form
\begin{equation}\label{ham}
H(\sigma)=
-J_p\sum\limits_{ \widetilde{>x,y<}}\sigma(x)\sigma(y) -J
\sum\limits_{<x,y>}\sigma(x)\sigma(y),
\end{equation}

Assume that  $\h: V\setminus\{x^{(0)}\}\to\ \mathbb{R}$ is a
mapping, i.e.
$$\h_{xy,\sigma(x)\sigma(y)}=(h_{xy,++},h_{xy,+-},h_{xy,-+},h_{xy,--}),$$
where
$\h_{xy,\sigma(x)\sigma(y)}\in \mathbb{R}$ ($\sigma(x),\sigma(y)\in \{-1, +1\}$) and $x,y\in V\setminus\{x^{(0)}\}$.

Now, we define Gibbs measure with memory of length 2 over Cayley tree.
\begin{equation}\label{mu}
\mu_{n,h}(\s)=\frac{1}{Z_{n}}\exp[-\beta H_n(\s)+\sum_{x\in
W_{n-1}}\sum_{y\in S(x)}\sigma(x)\sigma(y)\h_{xy,\sigma(x)\sigma(y)}].
\end{equation}

Here, as before, $\beta=\frac{1}{kT}$ and $\sigma_n: x\in V_n\to \sigma_n(x)$ and $Z_n$
is the corresponding to partition function
$$
Z_n=\sum\limits_{\sigma_n\in \Omega_{V_n}}\exp[-\beta H(\s_n)+\sum_{x\in
W_{n-1}}\sum_{y\in S(x)}\sigma(x)\sigma(y)\h_{xy,\sigma(x)\sigma(y)}].
$$

In this paper, we are interested in a construction of an
infinite volume distribution with given finite-dimensional
distributions. More exactly, we would like to find a
probability measure $\m$ on $\Om$ which is compatible
with given ones $\m_{\h}^{(n)}$, i.e.
\begin{equation}\label{CM}
\m(\s\in\Om: \s|_{V_n}=\s_n)=\m^{(n)}_{\h}(\s_n), \ \ \
\textrm{for all} \ \ \s_n\in\Om_{V_n}, \ n\in\bn.
\end{equation}
The consistency condition for $\mu_{n,h}(\s_n)$, $n\geq 1$ is
\begin{equation}\label{comp}
\sum_{\w\in\Om_{W_n}}\m^{(n)}_{\h}(\s_{n-1}\vee\w)=\m^{(n-1)}_{\h}(\s_{n-1}),
\end{equation}
for any $\s_{n-1}\in\Om_{V_{n-1}}$.

This condition according to the theorem implies the existence of a unique measure
$\m_{\h}$ defined on $\Om$ with a required condition \eqref{CM}. Such a measure $\m_{\h}$
is said to be {Gibbs measure} corresponding to the
model. Note that more general theory of measures has been developed in \cite{G}.

For $\s_{n-1}\in V_{n-1}$ and $\eta \in W_{n}$, we can define the Hamiltonian as following;

\begin{eqnarray}\label{ham1}
H_n(\s_{n-1}\vee\eta)
&=&-J\sum\limits_{<x,y>\in V_{n-1}}\sigma(x)\sigma(y) -J
\sum\limits_{x\in W_{n-1}}\sum\limits_{y\in S(x)}\sigma(x)\eta(y)\\\nonumber
&&-J_p\sum\limits_{>x,y<\in V_{n-1}}\sigma(x)\sigma(y) -J_p
\sum\limits_{x\in W_{n-2}}\sum\limits_{z\in S^2(x)}\sigma(x)\eta(z)\\\nonumber
&=&H_n(\s_{n-1})-J\sum\limits_{x\in W_{n-1}}\sum\limits_{y\in S(x)}\sigma(x)\eta(y)-J_p
\sum\limits_{x\in W_{n-2}}\sum\limits_{z\in S^2(x)}\sigma(x)\eta(z).
\end{eqnarray}
\begin{lem}\label{lemma1}
If $\frac{\mathbf{a}}{\mathbf{b}}=\frac{N_1}{N_2}$,
$\frac{\mathbf{a}}{\mathbf{c}}=\frac{N_1}{N_3}$ and $\frac{\mathbf{a}}{\mathbf{d}}=\frac{N_1}{N_4}$,
then there exists $D\in \mathbb{R}$ such that  $\mathbf{a}=D N_1$,  $\mathbf{b}=D N_2$,  $\mathbf{c}=D N_3$ and  $\mathbf{d}=D N_4$.
\end{lem}
\begin{proof}
Due to statement of Lemma, we have \\
$\mathbf{a}=D N_1$,  $\mathbf{b}=D N_2$,\\
$\mathbf{a}=D_1 N_1$,  $\mathbf{c}=D_1 N_3$,\\
 $\mathbf{a}=D_2 N_1$,  $\mathbf{d}=D_2 N_4$.
Therefore, from $\frac{\mathbf{a}}{\mathbf{b}}=\frac{D N_1}{D_1 N_2}=\frac{N_1}{N_2}$, $D=D_1$.
Similarly, from $\frac{\mathbf{a}}{\mathbf{d}}=\frac{D N_1}{D_2 N_4}=\frac{N_1}{N_4}$, $D=D_2$.
So, $D=D_1=D_2$.
\end{proof}

\begin{thm}\label{theorem1}
Probability distributions $\mu_n(\sigma_n)$, $n=1,2,...,$ in \eqref{mu} are compatible iff for any $x,y\in V$ the following equations hold:
\begin{eqnarray}\label{necessary}
e^{h_{xy,++}+h_{xy,-+}}&=&\prod\limits_{z\in S(y)}\frac{\exp[\h_{yz,++}](ab)^{2}+\exp[-\h_{yz,+-}]}
{\exp[\h_{yz,++}]a^2 +\exp[-\h_{yz,+-}] b^2}\\\nonumber
e^{h_{xy,--}+h_{xy,+-}}&=&\prod\limits_{z\in S(y)}\frac{\exp[-\h_{yz,-+}]+\exp[\h_{yz,--}](ab)^2}{\exp[-\h_{yz,-+}]b^2+\exp[\h_{yz,--}]a^2}\\\nonumber
e^{h_{xy,++}+h_{xy,+-}}&=&\prod\limits_{z\in S(y)}\frac{\exp[\h_{yz,++}](ab)^2+\exp[\h_{yz,+-}]]}{\exp[-\h_{yz,-+}]b^2+\exp[\h_{yz,--}]a^2},\nonumber
\end{eqnarray}
where $a=\exp(\beta J)$ and $b=\exp(\beta J_p)$.
\end{thm}
\begin{proof}
Necessity. From the equation \eqref{comp}, we have
\begin{eqnarray*}
&&L_n\sum\limits_{\eta\in \Omega_{W_{n}}}\exp[-\beta H_n(\s_{n-1}\vee \eta)+
\sum\limits_{x\in W_{n-1}}\sum\limits_{y\in S(x)}\sigma(x)\sigma(y)\h_{xy,\sigma(x)\sigma(y)}]\\
&=&\exp [-\beta H_n(\s_{n-1})+\sum\limits_{x\in W_{n-2}}\sum\limits_{y\in S(x)}\sigma(x)\sigma(y)\h_{xy,\sigma(x)\sigma(y)}],
\end{eqnarray*}
where $L_n=\frac{Z_{n-1}}{Z_n}$.

From the equation \eqref{ham1}, we have
\begin{eqnarray}\label{Kolmogorov1}
&&L_n\sum\limits_{\eta\in \Omega_{W_{n}}}\exp[-\beta H_n(\s_{n-1})-\beta J\sum\limits_{x\in W_{n-1}}\sum\limits_{y\in S(x)}\sigma(x)\eta(y)\\\nonumber
&-&\beta J_p\sum\limits_{x\in W_{n-2}}\sum\limits_{z\in S^2(x)}\sigma(x)\eta(z)+
\sum\limits_{x\in W_{n-1}}\sum\limits_{y\in S(x)}\sigma(x)\sigma(y)\h_{xy,\sigma(x)\sigma(y)}]\\\nonumber
&=&\exp [-\beta H_n(\s_{n-1})+\sum\limits_{x\in W_{n-2}}\sum\limits_{y\in S(x)}\sigma(x)\sigma(y)\h_{xy,\sigma(x)\sigma(y)}],
\end{eqnarray}
%

\begin{eqnarray*}\label{Kolmogorov2}
&&L_n\prod\limits_{x\in W_{n-2}}\prod\limits_{y\in S(x)}\prod\limits_{z\in S(y)}\sum\limits_{\eta(z)\in \{\mp 1\}}\exp[\sigma(y)\eta(z)\h_{yz,\sigma(y)\eta(z)}+\beta \eta(z)(J\sigma(y)+ J_p\sigma(x))]\\\nonumber
&=&\prod\limits_{x\in W_{n-2}}\prod\limits_{y\in S(x)}\exp [\sigma(x)\sigma(y)\h_{xy,\sigma(x)\sigma(y)}].
\end{eqnarray*}

Let us fix $<x,y>$ and considering all values of $\sigma(x), \sigma(y)\in \{-1,+1\}$. Then from \eqref{Kolmogorov1}, we have

\begin{equation}\label{Kolmogorov3}
e^{h_{xy,++}+h_{xy,-+}}=\prod\limits_{z\in S(y)}\frac{\sum\limits_{\eta(z)\in \{\mp 1\}}\exp[\eta(z)(\h_{yz,+\eta(z)}+\beta (J+ J_p))]}{\sum\limits_{\eta(z)\in \{\mp 1\}}\exp[\eta(z)(\h_{yz,+\eta(z)}+\beta (J-J_p))]}
\end{equation}
\begin{equation}\label{Kolmogorov4}
e^{h_{xy,--}+h_{xy,+-}}=\prod\limits_{z\in S(y)}\frac{\sum\limits_{\eta(z)\in \{\mp 1\}}\exp[-\eta(z)(\h_{yz,-\eta(z)}+\beta(J+ J_p))]}{\sum\limits_{\eta(z)\in \{\mp 1\}}\exp[-\eta(z)(\h_{yz,-\eta(z)}-\beta(-J+J_p))]}
\end{equation}

\begin{equation}\label{Kolmogorov4a}
e^{h_{xy,++}+h_{xy,+-}}=\prod\limits_{z\in S(y)}\frac{\sum\limits_{\eta(z)\in \{\mp 1\}}\exp[\eta(z)(\h_{yz,+\eta(z)}+\beta (J+ J_p))]}{\sum\limits_{\eta(z)\in \{\mp 1\}}\exp[-\eta(z)(\h_{yz,-\eta(z)}-\beta(-J+J_p))]}
\end{equation}
These equations imply the desired equations. This completes the proof.

Sufficiency: Now assume that the system of equations in \eqref{necessary} is valid, then from Lemma \ref{lemma1} we get

$$e^{\sigma(x)\sigma(y)h_{xy,\sigma(x)\sigma(y)}}D(x,y)=\prod\limits_{z\in S(y)}\sum\limits_{\eta(z)\in \{\mp 1\}}\exp[\sigma(y)\eta(z)\h_{yz,\sigma(y)\eta(z)}+\beta \eta(z)(J\sigma(y)+J_p \sigma(x))].$$
From the last equality, one can get
\begin{eqnarray}\label{Kolmogorov5}
&&\prod\limits_{x\in W_{n-2}}\prod\limits_{y\in S(x)}D(x,y)e^{\sigma(x)\sigma(y)h_{xy,\sigma(x)\sigma(y)}}\\\nonumber
&=&\prod\limits_{x\in W_{n-2}}\prod\limits_{y\in S(x)}\prod\limits_{z\in S(y)}\sum\limits_{\eta(z)\in \{\mp 1\}}e^{[\sigma(y)\eta(z)\h_{yz,\sigma(y)\eta(z)}+\beta \eta(z)(J\sigma(y)+ J_p \sigma(x))]}.
\end{eqnarray}
Multiply both sides of the equation \eqref{Kolmogorov5} by $e^{-\beta H_{n-1}(\sigma)}$ and denoting

$$U_n=\prod\limits_{x\in W_{n-2}}\prod\limits_{y\in S(x)}D(x,y),$$
we have from \eqref{Kolmogorov5},

\begin{eqnarray*}
&&U_{n-1}e^{-\beta H_{n-1}(\sigma)+\sum\limits_{x\in W_{n-2}}\sum\limits_{y\in S(x)}\sigma(x)\sigma(y)h_{xy,\sigma(x)\sigma(y)}}\\
&=&\prod\limits_{x\in W_{n-2}}\prod\limits_{y\in S(x)}\prod\limits_{z\in S(y)}e^{-\beta H_{n-1}(\sigma)}\sum\limits_{\eta(z)\in \{\mp 1\}}e^{[\sigma(y)\eta(z)\h_{yz,\sigma(y)\eta(z)}+\beta \eta(z)(J\sigma(y)+ J_p \sigma(x))]}.
\end{eqnarray*}
Therefore, we get
$$
U_{n-1}Z_{n-1}\mu_{n-1}(\sigma)=\sum\limits_{\eta}e^{-\beta H_{n}(\sigma\vee\eta)+\sum\limits_{x\in W_{n-2}}\sum\limits_{y\in S(x)}\sigma(x)\sigma(y)h_{xy,\sigma(x)\sigma(y)}}
$$

\begin{eqnarray}\label{eq4}
U_{n-1}Z_{n-1}\mu_{n-1}(\sigma)=Z_{n}\sum\limits_{\eta}\mu_{n}(\sigma\vee\eta).
\end{eqnarray}
As $\mu_n$ ($n\geq 1$) is a probability  measures, i.e.
$$
\sum\limits_{\sigma\in \{-1,+1\}^{V_{n-1}}}\mu_{n-1}(\sigma)=\sum\limits_{\sigma\in \{-1,+1\}^{V_{n-1}}}\sum\limits_{\eta \in \{-1,+1\}^{W_{n}}}\mu_{n}(\sigma\vee\eta)=1.
$$
From these equalities and the equation \eqref{eq4} we have $Z_{n}=U_{n-1}Z_{n-1}$.

This with the equation \eqref{eq4} implies that \eqref{comp} holds.
\end{proof}
According to Theorem \ref{theorem1} the problem of describing the Gibbs measures is reduced to the descriptions of the solutions of the functional equations \eqref{necessary}.

\section{The existence of phase transition}

Below we produce restrict ourselves to the case
\begin{eqnarray}\label{req1}
h_{xy,++}= h_{xy,-+}=h_1 \mbox{ and } h_{xy,--}= h_{xy,+-}=h_2.
\end{eqnarray}

Denoting Since $\ln u_1=h_{xy,++}= h_{xy,-+}$ and $\ln
u_2=h_{xy,--}= h_{xy,+-}$ for any $x,y\in V$, one can produce
\begin{equation}\label{newcase1}
u_1^2=\bigg(\frac{a^2 b^2 u_1u_2+1}{a^2u_1u_2+b^2}\bigg)^2=u_2^2,
\end{equation}
where $a=e^{\b J}$ and $b=e^{\b J_p}$.
Hence $u_1=u_z$. Let $u_1=u_2=u$, we have
$$
u=\frac{(a b)^2 u^2+1}{a^2 u^2+b^2}.$$ Let us consider a
non-linear function:
\begin{equation}\label{case1}
f(u)=\frac{(a b)^2 u^2+1}{a^2 u^2+b^2}.
\end{equation}
Let us substitute as $c=a^2$ and $d=b^2$, then we have a non-linear function as;
\begin{equation}\label{case11}
g(u)=\frac{cd u^2+1}{c u^2+d}.
\end{equation}

The analysis of the solution of the equations \eqref{necessary} is rather trick.
In this subsection, we will study the translation-invariant solutions, where $h_{xy}=h$ is constant $\forall x,y \in V.$ Let us consider \begin{equation}\label{case12}
u=\frac{cd u^2+1}{c u^2+d},
\end{equation}
where $u=e^h$. Note that if there is more than one positive solution for the equation \eqref{case12}, then we have more than one translation-invariant Gibbs measure corresponding to the solution of the equation \eqref{case12}. Fixed points $u$ correspond
to Markov chain Gibbs measures \cite{Zachary} (also Бsplitting Gibbs measures or Bethe Gibbs measures).

\begin{prop}
The equation $u=\frac{c d u^2+1}{c u^2+d}$ has one solution if either $c\leq 1$ or $d<3$. If $d\geq3$ then there exist $\eta_1(d)$,
$\eta_2(d)$ with $0<\eta_1(d)<\eta_2(d)$ such that equation \eqref{case12} has three solutions if $\eta_1(d)<c<\eta_2(d)$ and has two solutions if either
$\eta_1(d)=c$ or $\eta_2(d)=c$.
\end{prop}
\begin{proof} Denote \begin{equation*}\label{case11}
g(u)=\frac{cd u^2+1}{c u^2+d}.
\end{equation*}
We have
$$g'(u)=\frac{2 c (-1+d) (1+d) u}{\left(d+c u^2\right)^2}$$

$$g''(u)=-\frac{2 c (-1+d) (1+d) \left(-d+3 c u^2\right)}{\left(d+c u^2\right)^3}$$

In particular, if $d\leq 1$, then $g$ is decreasing and there can only be one solution of $g(u)=u$, therefore we can restrict ourselves to $d> 1$. We have that $g$ is convex for $0<u<\sqrt{\frac{d}{3c}}$ and is concave for $u>\sqrt{\frac{d}{3c}}$; thus there are at most three solutions to $g(u)=u$. In fact one can see that there is more than one solution if and only if there is more than one solution to $ug'(u)=u$, which is same as
$$
c^2 d u^4-c(d^2-3)u^2+d=0
$$
By means of the elementary analysis we get
$d^2-3>0$ and $\Delta=c^2(d^2-9)(d^2-1)>0$; and
$$u^2_1=\frac{(d^2-3)+\sqrt{(d^2-9)(d^2-1)}}{2cd}, \ \ u^2_2=\frac{(d^2-3)-\sqrt{(d^2-9)(d^2-1)}}{2cd}.
$$
From the elementary analysis, one can see that then $g'(u_1)<1<g'(u_2)$, if

\begin{eqnarray*}&&\frac{\left(-3+3 d^2+\sqrt{9-10 d^2+d^4}\right)^4}{32 d^3  (d^2-1)^2  \left(d^2-3+\sqrt{9-10 d^2+d^4}\right)}<c<\frac{\left(3-3 d^2+\sqrt{9-10 d^2+d^4}\right)^4}{32d^3  (d^2-1)^2  \left(d^2-3-\sqrt{9-10 d^2+d^4}\right)}.
\end{eqnarray*}
In this case, we have three fixed points of the function $g$.
\end{proof}
\begin{figure} [!htbp]%
 \centering
\includegraphics[width=65mm]{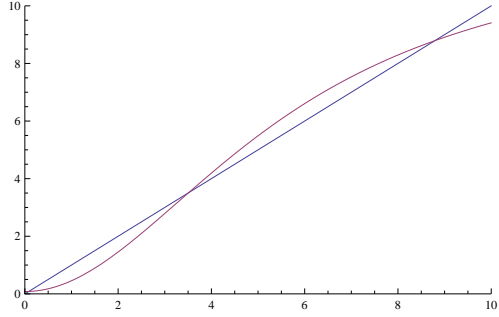}
\caption{$Jp = 34.6, J = -13., T = 27.5$} %
\end{figure}

\begin{figure} [!htbp]%
 \centering
\includegraphics[width=80mm]{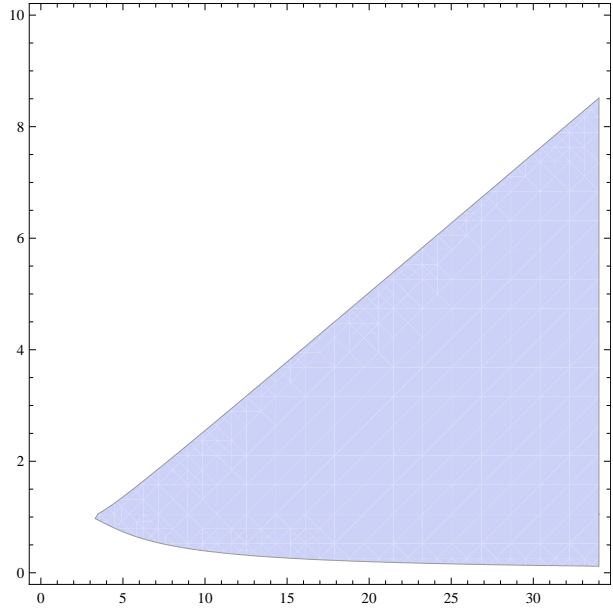}
\caption{The region with three positive fixed points of the function \eqref{case11}}
\end{figure}

\section{Free energy}
In this section, we study the free energy of depending on the boundary conditions for the Vannimenus Ising model on Cayley tree.
By previous sections we know that for any boundary condition satisfying the equations \eqref{necessary}
there exist Gibbs measures corresponding to the IV-model.
For this model the partition function is given by
\begin{eqnarray}\label{partition1}
Z_n=Z_n(\beta, h)=\sum\limits_{\s\in \Omega^{V_n}}\exp \big\{-\beta H_{n}(\s_n)+\sum\limits_{x\in W_{n-1}}\sum\limits_{y\in S(x)}\sigma(x)\sigma(y)h_{xy,\sigma(x)\sigma(y)}\big\}.
\end{eqnarray}
Here the spin configurations $\s_n$ belong to $\Omega^{V_n}$ and
$$
\textbf{h}=\{h_{xy,\sigma(x)\sigma(y)}\in \mathbb{R}, x,y\in V\}
$$
is a collection of real numbers that stands for boundary conditions.
In this section we will investigate the dependence with respect to boundary conditions  of the free energy defined as the limit:
\begin{eqnarray}\label{free1}
F(\beta, h)=\lim\limits_{n\to \infty }\frac{1}{\beta |V_n|}\ln Z_n(\beta, h).
\end{eqnarray}
Let us consider the equation
$$
D(x,y)e^{\sigma(x)\sigma(y)h_{xy,\sigma(x)\sigma(y)}}=\prod\limits_{z\in S(y)}\sum\limits_{u\in \{\pm 1\}}e^{u h_{yz,\sigma(y)u}+\beta u(J\sigma(y)+J_p\sigma(x))}
$$
$$
U_{n-1}=\prod\limits_{x\in W_{n-2}}\prod\limits_{z\in S(y)}\sum\limits_{u\in \{\pm 1\}}D(x,y).
$$
So, we have $Z_n=U_{n-1}Z_{n-1}.$

\subsection{The free energies of translation-invariant Gibbs measures}

In this section, we discuss the behavior of free energy as a
function of the model in presence of an external field. We as
before will consider the boundary conditions \eqref{req1}, i.e.
$$
h_{xy,++}= h_{xy,-+}=h_1 \mbox{ and } h_{xy,--}= h_{xy,+-}=h_2. \ \ \ \forall <x,y>\in L.
$$
\begin{prop}
The free energies of compatible translation-invariant (TI) boundary condition exist and are given by
\begin{eqnarray}\label{FE-TI1-case1}
F_{TI_{1}}(\beta, h)
&=&-\frac{\ln 2}{\beta}-\ln \big[\cosh(h_i+\beta(J+J_p))\cosh(h_i+\beta(J-J_p))\big],
\end{eqnarray}
where  $h_i$ are the varieties such that $u_i=e^{h_i}$ are the fixed points corresponding to the equations \eqref{necessary}.
\end{prop}
\begin{proof} From \eqref{req1} one finds
\begin{eqnarray}\label{10}\nonumber
D(x,y)e^{h_{xy,++}}&=&\prod\limits_{z\in S(y)}\big[e^{h_{yz,++}+\beta(J+J_p)}+e^{-h_{yz,+-}-\beta(J+J_p)}\big]\\
&=&\prod\limits_{z\in S(y)}2e^{\frac{h_{yz,++}-h_{yz,+-}}{2}}\cosh\big[\frac{h_{yz,++}-h_{yz,+-}}{2}+\beta(J+J_p)\big].
\end{eqnarray}

\begin{eqnarray}\label{11}\nonumber
D(x,y)e^{-h_{xy,-+}}&=&\prod\limits_{z\in S(y)}\big[e^{h_{yz,++}+\beta(J-J_p)}+e^{-h_{yz,+-}-\beta(J-J_p)}\big]\\
&=&\prod\limits_{z\in S(y)}2e^{\frac{h_{yz,++}-h_{yz,+-}}{2}}\cosh\big[\frac{h_{yz,++}-h_{yz,+-}}{2}+\beta(J-J_p)\big].
\end{eqnarray}
Multiply the equations \eqref{10} and \eqref{11}, then we have
\begin{eqnarray}\label{12}
D(x,y)
&=&4\prod\limits_{z\in S(y)}b(y,z),
\end{eqnarray}
where
$$b(y,z)=e^{\frac{h_{yz,++}-h_{yz,+-}}{2}}\big(\cosh\big[\frac{h_{yz,++}-h_{yz,+-}}{2}+\beta(J+J_p)\big]
\cosh\big[\frac{h_{yz,++}-h_{yz,+-}}{2}+\beta(J-J_p)\big]\big)^{\frac{1}{2}}.
$$

\begin{eqnarray}\label{13}\nonumber
U_{n-1}&=&\prod\limits_{x\in W_{n-2}}\prod\limits_{y\in S(x)}D(x,y)\\\nonumber
&=&4^{|W_{n-1}|}\prod\limits_{y\in W_{n-1}}\prod\limits_{z\in S(y)}b(x,y)=4^{|W_{n-1}|}e^{\sum\limits_{y\in W_{n-1}}\sum\limits_{z\in S(y)}\ln b(x,y)},
\end{eqnarray}
where $ \mathfrak{a}(x,y)=\ln \mathfrak{b}(x,y).$\\
For this model, we can give the partition functions as follows;
\begin{eqnarray}\label{14}\nonumber
Z_{n}&=&U_{n-1}Z_{n-1}\\
&=&4^{|V_{n-1}|}e^{\sum\limits_{y\in W_{n-1}}\sum\limits_{z\in S(y)}\mathfrak{a}(y,z)}e^{\sum\limits_{y_1\in W_{n-2}}\sum\limits_{z_1\in S(y)}\mathfrak{a}(y_1,z_1)}\ldots e^{\sum\limits_{\widetilde{y}\in W_{0}}\sum\limits_{\widetilde{z}\in S(\widetilde{y})}\mathfrak{a}(\widetilde{y},\widetilde{z})}\\\nonumber
&=&4^{|V_{n-1}|}e^{\sum\limits_{<x,y>\in V_{n}}\mathfrak{a}(x,y)}.
\end{eqnarray}
%
Therefore, we get the free energy as follows;
\begin{eqnarray}\label{FE-TI1}\nonumber
F_{TI_{1}}(\beta, h)&=&-\lim\limits_{n\to \infty }\frac{|V_{n-1}|}{\beta |V_n|}\ln D(x,y)\\
&=&-\frac{\ln 2}{\beta}-\ln \big[e^{h_1-h_2}\cosh(\frac{h_1+h_2}{2}+\beta(J+J_p))\cosh(\frac{h_1+h_2}{2}+\beta(J-J_p))\big].
\end{eqnarray}
Assume that $e^{h_1}=e^{h_1}$. According to the formula \eqref{FE-TI1}, for the case 1 of Vannimenus Ising model the free energies
of translation-invariant (TI) boundary condition are given by

\begin{eqnarray}\label{FE-TI1-case1}
F_{TI_{1}}(\beta, h)
&=&-\frac{\ln 2}{\beta}-\ln \big[\cosh(h_i+\beta(J+J_p))\cosh(h_i+\beta(J-J_p))\big],
\end{eqnarray}
where  $h_i$ are the varieties such that $u_i=e^{h_i}$ are the fixed points corresponding to three cases.
\end{proof}

\begin{figure} [!htbp]%
 \centering
\includegraphics[width=80mm]{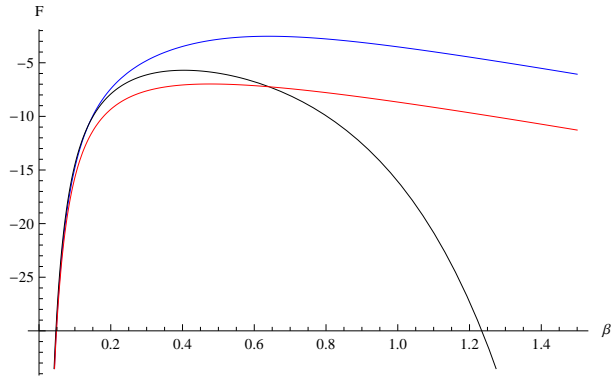}
\caption{$J_p = 34.6, J = -13, u_1 =0.260261, u_2 = 1.18483, u_3 = 3.52491$} %
\end{figure}

The residual entropy at $T=0$ is defined as follows:

\begin{eqnarray}\label{FE-TI1-case1}
S_{\infty}&=&-\lim\limits_{\beta\to \infty }\frac{F(h)-F_{\infty}}{\frac{1}{\beta}},\\\nonumber
\end{eqnarray}
where $F_{\infty}=\lim\limits_{\beta\to \infty }F(h).$

Let us compute the entropy
\begin{eqnarray}\label{FE-TI1-case1}
S(\beta,h)&=&-\frac{d F(\beta,h)}{dT}=\frac{d F(\beta,h)}{d\beta}\frac{1}{\beta^2}\\\nonumber
&=&\frac{\ln 2}{\beta^2}-(J+J_p)\tanh[h + \beta(J + J_p)]-(J - J_p)\tanh[h+\beta(J - J_p)]
\end{eqnarray}



The residual entropy at $T=0$ is defined as follows:

\begin{eqnarray}\label{FE-TI3-case3}
S_{\infty}&=&-\lim\limits_{\beta\to \infty }\frac{F(h)-F_{\infty}}{\frac{1}{\beta}},\\\nonumber
\end{eqnarray}
where $F_{\infty}=\lim\limits_{\beta\to \infty }F(h).$

Let us compute the entropy
\begin{eqnarray}\label{FE-TI1-case3}
S(\beta,x_i)&=&-\frac{d F(\beta,h)}{dT}=\frac{d F(\beta,h)}{d\beta}\frac{1}{\beta^2}\\\nonumber
&=&\frac{\ln 2}{\beta^2}-(J+J_p)\tanh[\frac{\ln x_i}{2} + \beta(J + J_p)]+(J+J_p)\tanh[\frac{\ln x_i}{2}-\beta(J+J_p)]
\end{eqnarray}

\noindent {\bf Acknowledgements} The second named author (F.M.) thanks The Scientific and Technological Research Council of
Turkey-TUBITAK for providing financial support and Zirve University for kind hospitality and providing all facilities. 

\begin{thebibliography}{99}
\bibitem{Iof} Ioffe, D.:
A note on the extremality of the disordered state for the Ising model on the Bethe lattice
Lett. Math. Phys., \textbf{37}, 137-143 (1996).
\end{thebibliography}
\end{document}